\documentclass{article}

\usepackage{amsmath,amsthm,amssymb,cite,enumerate} 

\usepackage{color,graphicx}

\newtheorem{thm}{Theorem}[section] 
\newtheorem{cor}[thm]{Corollary} 
\newtheorem{lem}[thm]{Lemma} 
\newtheorem{prop}[thm]{Proposition}
 
\theoremstyle{definition} 
\newtheorem{defn}[thm]{Definition}
  
\theoremstyle{remark}  
  
\def\beq{\begin{eqnarray}}  
\def\eeq{\end{eqnarray}}  
\def\bsp{\begin{split}}  
\def\esp{\end{split}}

\def\d{\mathrm{d}}

\def\T{ {\sf T} }

\def\CC{CSI$_{\text{cons}}$}

\newcommand{\pd}[2]{\frac{\partial {#1}}{\partial {#2}}}

\def \bl {\mbox{{\mbold\ell}}}
\def \bn {\mbox{{\bf n}}}

\def \sh{\mathrm{sinh} pw} 
\def \ch{\mathrm{cosh} pw} 

\def \BEA { \begin{eqnarray}}
\def \EEA {\end{eqnarray}}
\def \BE {\begin{equation}}
\def \EE {\end{equation}}
\def\d{\mathrm{d}}

\def \mi {\stackrel{i}{m}}
\def \mj {\stackrel{j}{m}}
\def \mk {\stackrel{k}{m}}
\def \mr {\stackrel{r}{m}}
\def \ms {\stackrel{s}{m}}
\def \mz {\stackrel{z}{m}}
\def \mq {\stackrel{q}{m}}
\def \mo {\stackrel{o}{m}}
\def \mD {\stackrel{2}{m}}
\def \mT {\stackrel{3}{m}}
\def \mC {\stackrel{4}{m}}

\def \mio #1 {\mi_{#1}\ ^{  \! \! \! \! 0}} 
\def \mjo #1 {\mj_{#1}\ ^{  \! \! \! \! 0}} 
\def \mko #1 {\mk_{#1}\ ^{  \! \! \! \! 0}} 
\def \mro #1 {\mr_{#1}\ ^{  \! \! \! \! 0}} 
\def \mso #1 {\ms_{#1}\ ^{  \! \! \! \! 0}} 
\def \mpo #1 {\mp_{#1}\ ^{  \! \! \! \! 0}} 
\def \mzo #1 {\mz_{#1}\ ^{  \! \! \! \! 0}} 
\def \mqo #1 {\mq_{#1}\ ^{  \! \! \! \! 0}} 
\def \moo #1 {\mo_{#1}\ ^{  \! \! \! \! 0}} 
\def \mDo #1 {\mD_{#1}\ ^{  \! \! \! \! 0}} 
\def \mTo #1 {\mT_{#1}\ ^{  \! \! \! \! 0}} 
\def \mCo #1 {\mC_{#1}\ ^{  \! \! \! \! 0}} 

\def \miJ #1 {\mi_{#1}\ ^{  \! \! \! \! (1)}} 
\def \mjJ #1 {\mj_{#1}\ ^{  \! \! \! \! (1)}} 
\def \mkJ #1 {\mk_{#1}\ ^{  \! \! \! \! (1)}} 
\def \mrJ #1 {\mr_{#1}\ ^{  \! \! \! \! (1)}}

\def \bl {\mbox{\boldmath{$\ell$}}}
\def \hbl {\mbox{\boldmath{$\hat \ell$}}}
\def \bn {\mbox{\boldmath{$n$}}}

\def \hbn {\mbox{\boldmath{$\hat n$}}}
\def \hbm #1 {\mbox{\boldmath{$\hat m^{(#1)}$}}}

\def \Mi {\stackrel{i}{M}}

\def \Mk {\stackrel{k}{M}}

\newcommand{\be}{\begin{equation}}
\newcommand{\ee}{\end{equation}}
\newcommand{\beqn}{\begin{eqnarray}}
\newcommand{\eeqn}{\end{eqnarray}}

\newcommand{\ba}{\begin{array}}
\newcommand{\ea}{\end{array}}
\newcommand{\pp}{{\it pp\,}-}

\def \Mi {\stackrel{i}{M}}

\def \Mk {\stackrel{k}{M}}

\def \pul {{{\footnotesize{\frac{1}{2}}}}}

\def \T {\bigtriangleup  }

\newcommand{\M}[3] {{\stackrel{#1}{M}}_{{#2}{#3}}}

\begin{document}

\title{Type III and N universal spacetimes }

\author{Sigbj\o rn Hervik$^\diamond$, Vojt\v ech Pravda$^\star$,  Alena Pravdov\' a$^\star$\\
\vspace{0.05cm} \\
{\small $^\diamond$ Faculty of Science and Technology, University of Stavanger}, {\small  N-4036 Stavanger, Norway}  \\
{\small $^\star$ Institute of Mathematics, Academy of Sciences of the Czech Republic}, \\ {\small \v Zitn\' a 25, 115 67 Prague 1, Czech Republic} \\
 {\small E-mail: \texttt{sigbjorn.hervik@uis.no,
pravda@math.cas.cz, pravdova@math.cas.cz}} }

\date{\today}

\maketitle

\begin{abstract}
Universal spacetimes are spacetimes for which  all conserved symmetric rank-2 tensors, constructed   as contractions of polynomials from  the  metric, the Riemann tensor and its covariant  derivatives
of arbitrary order, are multiples of the metric. Consequently,  metrics   of universal spacetimes solve vacuum equations of all gravitational theories with Lagrangian    being  a polynomial curvature invariant constructed  from the metric, the Riemann tensor and its derivatives of arbitrary order. 
In the literature, universal metrics are also discussed  as metrics with vanishing quantum corrections and as classical solutions to string theory. Widely known examples of universal metrics are certain Ricci-flat pp waves.
 
In this paper, we start a  general study of geometric properties of universal metrics in arbitrary dimension and we arrive at a broader class of such metrics. In contrast with pp waves, these universal metrics also admit non-vanishing cosmological constant and in general do not have to  possess a covariantly constant or recurrent  null vector field.
First, we show that a universal  spacetime is necessarily a CSI spacetime, i.e. all curvature invariants constructed from the Riemann tensor and its derivatives are constant.  Then we focus on  type N spacetimes, where we arrive at a  simple necessary and sufficient condition:
 a type N spacetime is universal if and only if it is an Einstein Kundt spacetime.  A class of type III Kundt  universal metrics is also found.
  Several explicit examples of universal metrics are presented.

\end{abstract}

\section{Introduction}

Already in late 80's and early 90's, it was shown  \cite{AmaKli89,HorSte90} that certain Ricci-flat pp waves are solutions to all gravitational theories with the Lagrangian constructed  from the metric, the Riemann tensor and its derivatives of arbitrary order, thus being  also classical solutions e.g. to string theory. A natural question  
arises  whether there exist other spacetimes with this property.
This question was answered affirmatively in \cite{Coleyetal08} where {\em universal} metrics were discussed as metrics with vanishing quantum corrections and where  further examples  of  { universal} metrics were given. In fact,  the term { universal} was  also introduced in \cite{Coleyetal08} 
\begin{defn}
\label{univ}
A metric is called {\it universal} if all conserved symmetric rank-2 tensors constructed  
from  the  metric, the Riemann tensor and its covariant  derivatives
of arbitrary order\footnote{{  Throughout the paper, we consider only scalars and symmetric rank-2 tensors  constructed as contractions of {\em polynomials} 
from  the  metric, the Riemann tensor and its covariant  derivatives
of arbitrary order. }} are multiples of the metric.
\end{defn}
Thus, a universal metric  has $T_{ab}=\lambda g_{ab}$ {  ($\lambda$ is constant)} for any conserved symmetric rank-2 tensor  $T_{ab}$ constructed from the metric, the Riemann tensor and its covariant derivatives. 
Such metrics are thus vacuum solutions (possibly with a non-vanishing cosmological constant) to all theories  with the Lagrangian {  being a polynomial curvature invariant}\footnote{  {In fact universal metrics also solve field equations derived from the Lagrangian that is  an analytic function of such polynomial invariants (e.g.   $f(R)$ gravity).}} in the form
\be
{ L}={ L}(g_{ab},R_{abcd},\nabla_{a_1}R_{bcde},\dots,\nabla_{a_1\dots a_p}R_{bcde}),\label{Lagr}
\ee
 which are natural geometric generalizations of the Einstein gravity {  (e.g. the Lovelock gravity, the quadratic gravity, etc.)}.

In the Riemannian signature, these metrics have been studied by Bleecker \cite{Bleecker}, who has found that all such metrics are necessarily isotropy-irreducible homogeneous metrics. This is also a sufficient condition. In particular, this means that all polynomial invariants constructed from the Riemann tensor and its derivatives of arbitrary order are constants. 

In this paper, we are concerned with the Lorentzian case and we study general properties of universal spacetimes in an arbitrary dimension and we identify various classes of these spacetimes.

First, we employ  a necessary condition for universal spacetimes derived in section \ref{sec_univCSI}.  
Since $T^a_{~a}= \text{constant}$ for a universal spacetime, one can prove (see the proof of theorem \ref{CSIconst})
\begin{thm}
\label{prop_univCSI}
{ A universal  spacetime is necessarily a CSI spacetime.} 
\end{thm}
CSI {  (constant curvature invariant)} spacetimes  are spacetimes for which  all curvature invariants constructed from the metric, the Riemann tensor and its derivatives of arbitrary order are constant, see e.g. \cite{ColHerPel06}. Therefore, the CSI spacetimes  play an important role in the question of universality and as potential solutions of all theories of gravity {  with the Lagrangian of the form \eqref{Lagr}}. Note, however, that the CSI criterion is not sufficient {  for universality}; on the other hand, there is an interesting reverse result (see section \ref{sec_univCSI}): 
\emph{Given a CSI metric, $g_{ab}$, then there exists a class of theories having $g_{ab}$ as a vacuum solution.}

As far as universal spacetimes are concerned, the
CSI constraint on universal spacetimes from theorem \ref{prop_univCSI}
allows us to systematically study necessary and sufficient conditions for universal metrics of types N and III in the algebraic classification
of  \cite{Coleyetal04} ({  see also eqs. \eqref{typeNWeyl} and \eqref{typeIIIWeyl} and } \cite{OrtPraPra12rev} for a recent review).

In particular, for type N,\footnote{Note that while in \cite{Coleyetal08}, it was already mentioned  without proof that
type N Einstein   Kundt  spacetimes are universal, to our knowledge necessary conditions have been  never studied. } we derive the following necessary and sufficient condition: 
\begin{thm}
\label{prop_typeN}
{ A type N spacetime is universal if and only if it is an Einstein\footnote{  Einstein spacetimes are spacetimes obeying $R_{ab}= (R/n) g_{ab}$.} Kundt spacetime.}
\end{thm}
Kundt spacetimes are spacetimes admitting non-expanding, shearfree, twistfree  geodetic null congruence (see section \ref{sec_Kundt} for the discussion of Kundt spacetimes).

Universal metrics  of \cite{AmaKli89,HorSte90}  belong to a more general class of $\tau_i=0$, type N Ricci-flat Kundt metrics {  (see eq.  \eqref{dl} for the definition of $\tau_i$)}. Furthermore, type N Einstein  Kundt metrics also allow for
 $\tau_i \not=0$ and for non-vanishing cosmological constant $\Lambda$ (see section \ref{sec_Kundt} for corresponding metrics).

For type III spacetimes, the Einstein Kundt condition is not a sufficient condition for a spacetime being universal since further 
constraints follow. In this case, we prove the following sufficient conditions. 

\begin{thm}
\label{prop_typeIII}
Type III, $\tau_i=0$ Einstein Kundt spacetimes obeying  
\BE
C_{acde} C_{b}^{\ cde}=0 \label{QGterm} 
\EE
    are universal. 
\end{thm}

{  {The above theorem states sufficient conditions  for type III universal spacetimes. For necessary conditions we have only partial results. We prove that type III universal spacetimes belong to the Kundt class only under some genericity assumptions (see section
\ref{sec_necessaryIII}). 
 We also show that the condition \eqref{QGterm} 
 is a necessary condition for universality of type III spacetimes  (section \ref{sec_typeIII}).}}

While in four dimensions, pp waves (i.e. spacetimes admitting covariantly constant null vector - CCNV) are of type N, in five dimensions, they are of type N or III and in dimensions $n>5$, pp waves of type N, III and II  exist
\cite{OrtPraPra12rev}. Theorem \ref{prop_typeN} implies that type N pp waves are universal. On the other hand, type III pp waves are
not compatible with the necessary condition for type III  universal spacetimes \eqref{QGterm} and therefore are not universal (see section \ref{sec_Kundt}).

{ 
This paper is organized as follows: in section \ref{sec_prelim}, we briefly introduce the notation and summarize some
useful definitions and results used in this paper. 

In section \ref{sec_univCSI}, we prove that universal spacetimes are CSI (theorem \ref{prop_univCSI}). However, note  that the CSI condition is not sufficient for universality. Note also that, in contrast with the rest of the paper, in this section, we do not assume a type N nor type III Weyl tensor.

In section \ref{sec_typeN}, we prove the necessary and sufficient conditions for type N universal spacetimes (theorem \ref{prop_typeN}). Similar methods as in  section \ref{sec_typeN} for type N spacetimes
 are also used in section \ref{sec_typeIII} for the type III case  in a more complicated proof of theorem \ref{prop_typeIII}. 

Section \ref{sec_Kundt} is devoted to explicit examples of universal metrics. In section \ref{sec_concl}, we briefly summarize and discuss the main results. Finally, in appendix \ref{App:NPeqs}, we give a subset of the Bianchi and Ricci identities used in this paper.}

\section{Preliminaries}
\label{sec_prelim}

In this paper, we employ the algebraic classification of the Weyl tensor \cite{Coleyetal04}  and higher dimensional generalizations of the Newman-Penrose  \cite{Pravdaetal04, OrtPraPra07}   and the Geroch-Held-Penrose formalisms 
\cite{Durkeeetal10}. We  follow the notation of \cite{Durkeeetal10,OrtPraPra12rev}. Here, we only briefly summarize definitions and results necessary in this paper 
(for more thorough introduction see  the recent review \cite{OrtPraPra12rev}).

In an $n$-dimensional Lorentzian manifold, we  work in a null frame with two null vectors $\bl$ and $\bn$ and $n-2$ spacelike vectors $\mbox{\boldmath{$m^{(i)}$}} $ obeying\footnote{Hereafter, coordinate indices $a,b, \ldots $ and frame indices $i,j,\ldots$ take values from 0 to $n-1$ and 2 to $n-1$, respectively.}
\be
\ell^a \ell_a= n^a n_a = 0, \qquad   \ell^a n_a = 1, \qquad \ m^{(i)a}m^{(j)}_a=\delta_{ij}.  \label{ortbasis}
\ee
The Lorentz transformations between null frames are generated by boosts
\be
\hbl = \lambda \bl, \qquad    \hbn = \lambda^{-1} \bn, \qquad   \hbm{i} = \mbox{\boldmath{$m^{(i)}$}},  \label{boost}
\ee
null rotations and spins.
We say that quantity  $q$  has  a {boost weight} (b.w.)   ${\rm b}$  if it transforms under a boost~\eqref{boost} according to
\be
\hat q = \lambda^{\rm b} q . 
\ee 
Obviously, various components of a tensor in a null frame  may have distinct integer boost weights. For our purposes, it is convenient
to introduce as {{boost order}} of a tensor with respect to a null frame the maximum boost weight of its frame components taken over non-vanishing components.
Note that in fact, the boost order of a tensor depends only on the null direction $\bl$ {  (see Proposition 2.1 in \cite{OrtPraPra12rev})}.

As type N spacetimes are defined those spacetimes admitting a frame in which only  boost-weight -2 components of the Weyl tensor are present.
Consequently, in an appropriately chosen frame, the Weyl tensor can be expressed as \cite{Coleyetal04,OrtPraPra12rev}
\be
C_{abcd} = 4 {\Omega'}_{ij}\, \ell^{}_{\{a} m^{(i)}_{\, b}  \ell^{}_{c}  m^{(j)}_{\, d\: \}}, \label{typeNWeyl}
\ee
where {  for an arbitrary tensor $T_{abcd}$}
\be
T_{\{ a bc d\} } \equiv \pul (T_{[a b] [c d]}+ T_{[c d] [ab]})  
\label{Weyl_symm}
\ee
{  (and so $C_{abcd}=C_{\{ abcd\}}$)},  ${\Omega'}_{ij}$ is symmetric and traceless {  and $\bl$ is the multiple WAND}. For type N Einstein  spacetimes, $\bl$ is necessarily geodetic \cite{Pravdaetal04}.
Without loss of generality, we can choose $\bl$ to be affinely parameterized and a frame parallelly transported along $\bl$.
Then, the covariant derivatives of the frame vectors in terms of spin coefficients read  
\BEA
\ell_{a ; b} &=& L_{11} \ell_a \ell_b  + L_{1i} \ell_a m^{(i)}_{\, b}  +
\tau_i 
m^{(i)}_a \ell_b    + {\rho}_{ij} m^{(i)}_{\, a} m^{(j)}_{\, b}  , \label{dl} \\
n_{a ; b  } &=&\! -\! L_{11} n_a \ell_b -\! L_{1i} n_a m^{(i)}_{\, b}  +
\kappa'_i 
 m^{(i)}_{\, a} \ell_b   + \rho'_{ij} 
m^{(i)}_{\, a} m^{(j)}_{\, b}  , \label{dn} \\
m^{(i)}_{a ; b } &=&\! -\! \kappa'_i 
\ell_a \ell_b  -\tau_i 
n_a \ell_b  
-\rho'_{ij} 
\ell_a m^{(j)}_{\, b}   
 \! +\! {\Mi}_{j1} m^{(j)}_{\, a} \ell_b  -\rho_{ij} n_a m^{(j)}_{\, b} 
+ {\Mi}_{kl} m^{(k)}_{\, a} m^{(l)}_{\, b} . \label{dm} 
\EEA
It is often convenient to decompose the optical matrix $\rho_{ij}$ 
into its trace $\theta$ (expansion), tracefree symmetric part $\sigma_{ij}$ and antisymmetric part $A_{ij}$
\BEA
 \rho_{ij}=\sigma_{ij}+\theta\delta_{ij}+A_{ij}, \label{L_decomp} \qquad 
 \sigma_{ij}\equiv \rho_{(ij)}-\textstyle{\frac{1}{n-2}} \rho_{kk}\delta_{ij}, 
\qquad \theta\equiv\textstyle{\frac{1}{n-2}}{\rho_{kk}}, \qquad A_{ij}\equiv \rho_{[ij]}. \label{opt_matrices}
\EEA
 Shear and {twist} of $\bl$ correspond to the traces $\sigma^2\equiv 
\sigma_{ij}\sigma_{ji}$ and $\omega^2\equiv 
-A_{ij}A_{ji}$,
respectively. {  Note that since indices $i,j,\dots$ are raised and lowered using $\delta_{ij}$ we do not distinguish between their covariant and contravariant versions and thus e.g. $\rho_{kk}$ is equivalent to ${\rho^k}_k$.}

It can be shown that, for algebraically special Einstein spacetimes, possible forms of $\rho_{ij}$ are rather restricted. In particular, for type N Einstein spacetimes, a parallelly propagated frame can be chosen so that \cite{Pravdaetal04,Durkeeetal10,OrtPraPra10}
\BEA
 \rho_{ij}=  {\mbox{diag}} \left( \left[ \begin{array}{cc} S & A   \\ -A & S   \end{array} \right], 0, \dots, 0 \right) . \label{GStypeN}
\EEA

Covariant derivatives in directions of the frame vectors $\bl$, $\bn$ and $\mbox{\boldmath{$m^{(i)}$}}$  will be denoted as $D$, $\bigtriangleup $
and $\delta_i$, respectively, so that
\BE
\nabla_a = n_a D + m^{(i)}_{\, a} \delta_i + \ell_a \bigtriangleup.   \label{eqnabla}
\EE

{  For Einstein spacetimes, the tensor
\BE 
S_{ab} \equiv C_{acde} C_{b}^{\ cde}-\frac{1}{4} g_{ab} C_{cdef} C^{cdef} \label{conserved}
 \EE
is conserved since
\BE
S^{ab}_{\ \ ;b} = C^a_{\ cde;b} C^{bcde} - \frac{1}{2} C_{cdef;}^{\ \ \ \ \ a} C^{cdef} 
\EE
identically vanishes due to the Bianchi identities. In fact, the conservation of \eqref{conserved} also directly follows from the field equations of the quadratic gravity for Einstein spacetimes \cite{MalekPravdaQG}, which are derived from the diffeomorphism-invariant Lagrangian and are therefore conserved. Note that for type III Einstein spacetimes, the second term in  \eqref{conserved} vanishes and therefore in such a case
\BE
 S^{{\rm III}}_{ab} \equiv C_{acde} C_{b}^{\ cde} \label{typeIIIconserved}
\EE
 is also conserved, while for type N Einstein spacetimes, both terms in \eqref{conserved} vanish.
} 
   

\section{Conserved tensors}
\label{sec_univCSI}

Let us consider Lagrangians containing the Riemann tensor and its derivatives up to a fixed order $p$. Then it is shown in \cite{IYER-WALD} that it can be rewritten in the form
\beq
L=L(g_{ab},R_{abcd},\nabla_{a_1}R_{bcde},...,\nabla_{(a_1...a_p)}R_{bcde}).
\eeq
By varying the action, 
one obtains \cite{IYER-WALD}
\beq
-T^{ab}&=&\pd{L}{g_{ab}}+E^a_{~cde}R^{bcde}+2\nabla_c\nabla_dE^{acdb}+\frac 12g^{ab} L, \\
E^{bcde} &=& \pd{L}{R_{bcde}}-\nabla_{a_1}\pd{L}{\nabla_{a_1}R_{bcde}}+\cdots+(-1)^p\nabla_{(a_1}\cdots\nabla_{a_p)}\pd{L}{\nabla_{(a_1}\cdots\nabla_{a_p)}R_{bcde}}.\nonumber
\label{eq18}
\eeq
Here, the tensor $T^{ab}$ is the associated conserved tensor\footnote{{  Since the action is diffeomorphism invariant, a variation of the action will give a conserved tensor $T^{ab}$ irrespective whether the field equations are obeyed or not. In Appendix  in \cite{IYER-WALD}, they consider a non-dynamical metric where the field equations are not necessarily satisfied (RHS in eq.(\ref{eq18}) is not necessarily zero), only the matter equations are satisfied. In our case, no matter fields are present, and hence, the matter equations are trivially satisfied implying that $\nabla_aT^{ab}=0$ regardless whether the field equations for $g_{ab}$ are satisfied or not. We should therefore interpret the conserved property of $T^{ab}$ as an identity (being a result of the Bianchi identities, Ricci identities, etc.) rather than an additional condition.

See also Lemma 2.5  in \cite{Bleecker} and its proof - the metric variation (the divergence in the space of metrics) of an invariant function $f$, e.g. the integral of an invariant scalar, is conserved, i.e.  divergence free w.r.t. the metric.
}}. By taking the trace, we get
\beq
-T^a_{~a}=g_{ab}\pd{L}{g_{ab}}+E_{bcde}R^{bcde}+2\nabla_c\nabla_dE_a^{~cda}+\frac D2 L,
\label{tracefeq}\eeq
where $D$ is the dimension of the spacetime.

\begin{defn}
Let \CC\ be the set of metrics with the following property:
For all conserved symmetric rank-2 tensors, $T_{ab}$, then the traces $T^a_{~a}$ are constant. 
\end{defn}
Then:
\begin{thm}
\label{CSIconst}
A metric is {\rm \CC}   if and only if it is {\rm CSI}; i.e., {\rm \CC}\  $\Leftrightarrow$ {\rm CSI}. 
\end{thm}
\begin{proof}
{  Since the traces $T^a_{~a}$ are particular kinds of polynomial invariants}, CSI$\Rightarrow$ \CC\ is obvious. We thus need to prove that \CC $\Rightarrow$ CSI. Assume therefore that a metric is \CC. 

Consider any polynomial invariant $I$. Assume that the invariant contains derivatives of the Riemann tensor of orders at most $p$ and by  \cite{IYER-WALD} we can assume it is of the form
\[ I=I[g_{ab},R_{abcd},\nabla_{a_1}R_{bcde},...,\nabla_{(a_1...a_p)}R_{bcde}]. \]
Let us consider the (infinite) series of Lagrangians $L=I^n$, $n=1,2,3,...$. For each $n$, we get, by variation, a conserved tensor $T[n]^a_{~b}$. Consider then the trace of these tensors, which {  by the \CC\ assumption} are all  constants, $-T[n]^a_{~a}=c_n$. 

We can now use the equations for the variation $(\ref{tracefeq})$ to obtain an equation for each $n$. For $n=1$, we get
\beq
\frac{D}{2}I+g_{ab}\pd{I}{g_{ab}}+\tilde{E}_{bcde}R^{bcde}+2\nabla_c\nabla_d\tilde{E}_a^{~cda}=c_1, 
\eeq
where $\tilde{E}_{bcde}$ is $E_{bcde}$ with $L$ fixed as $L=I$. We note that the sum of the 2nd, 3rd, and 4th terms on the LHS is an invariant itself, which we conveniently define as
\beq
X_1\equiv g_{ab}\pd{I}{g_{ab}}+\tilde{E}_{bcde}R^{bcde}+2\nabla_c\nabla_d\tilde{E}_a^{~cda}, 
\eeq
so that 
\[ \frac{D}{2}I+X_1=c_1.\]
{  In order to get an equation for every $n$, we consider the more general case $L=f(I)$, of which $L=I^n$ are special cases.  When inserting $L=f(I)$ into equation (\ref{tracefeq}), which contains maximum $(p+3)$ derivatives of various sorts, we get an expression of the form: 

\beq \frac{D}{2}f(I)+f'(I)X_1+f''(I)X_2+...+f^{(j)}(I)X_j,
\label{InFE}
\eeq
where $X_i$ are polynomial invariants constructed from (differentials) of $I$. This can be seen as follows. The terms in (\ref{tracefeq}) contain derivatives of $L$ of various sorts, symbolically called $\pd{L}{x}$. We then get
\beq
\pd{f(I)}{x}&=&f'(I)\pd{I}{x}, \nonumber \\
\nabla_{a}\pd{f(I)}{x}&=&f'(I)\nabla_a\pd{I}{x}+f''(I)(\nabla_aI)\pd{I}{x}, \nonumber \\
\nabla_b\nabla_{a}\pd{f(I)}{x}&=&f'(I)\nabla_b\nabla_a\pd{I}{x}+f''(I)\left[2\left(\nabla_{(a}I\right)\nabla_{b)}\pd{I}{x}+(\nabla_b\nabla_aI)\pd{I}{x}\right] 
  +f'''(I)(\nabla_aI)(\nabla_b I)\pd{I}{x}.
\eeq 
In general, the derivatives will be of the form (symbolically)
\[
\nabla^{(j)}\pd{f(I)}{x}=f'(I)D_1[I]+f''(I)D_2[I]+...+f^{(j+1)}(I)D_{j+1}[I],
\]
 where $D_i[I]$ are some derivative operators which do not depend on the functional form of $f$ (on $I$ only). Thus, in (\ref{tracefeq}) we collect all terms proportional to $f^{(i)}(I)$ so in (\ref{tracefeq}) we get a term
\[ f^{(i)}(I) X_i, \] 
where $X_i$ is independent of the functional form of $f$. Hence, we arrive at expression (\ref{InFE}). Note, however, since the expression \eqref{tracefeq} only contains a maximum of $(p+3)$ derivatives we have that if $i>p+3$, then $X_i=0$; consequently, there is only a finite number of possible non-zero $X_i$, $1\leq i\leq p+3$. 

We are considering the polynomial cases $f(I)=I^n$ which lead to a simplification. For example, considering $n=2$, and imposing the the \CC\ condition, we get the equation:
\[  \frac{D}{2}I^2+2IX_1+2X_2=c_2, \]
where $X_1$ is defined above, and $X_2$ is also a polynomial invariant. In general, for any $n$, we thus get the equations: 
\beq \frac{D}{2}I^n+nI^{n-1}X_1+n(n-1)I^{n-2}X_2+...+n!X_n=c_n,
\eeq
where $X_i$ are polynomial invariants.\footnote{ {For example, for $p=0$,  
$X_2\equiv 2\left( \nabla_c I \nabla_d {\tilde{E}^{acd}}_a+\nabla_d I \nabla_c {\tilde E}_a^{~cda}+\tilde{E}_a^{~cda}\nabla_c\nabla_d I\right)$, 
$X_3\equiv 2{\tilde E}_a^{~cda} \nabla_c I \nabla_d I$
.}}
}
We can now write down a set of $k+1=p+4$ equations in a  matrix form
\[ {\bf A}{\bf x}={\bf c} ,\] 
where ${\bf x}$ and ${\bf c}$ are the column vectors ${\bf x}=(D/2,X_1,X_2,...,X_k)^T$ and ${\bf c}=(c_1,c_2,c_3,...,c_{k+1})^T$, and ${\bf A}$ is the matrix
\beq
{\bf A}=\begin{bmatrix}
I & 1 & 0 & 0 & \cdots & 0 \\
I^2 & 2I & 2 & 0 & \cdots & 0 \\
I^3 & 3I^2 & 3\cdot2 I & 3\cdot2\cdot 1 &\cdots & 0 \\
\vdots & \vdots & \vdots  & \vdots & & \vdots  \\
I^k & kI^{k-1} & k(k-1)I^{k-2} &  k(k-1)(k-2)I^{k-3}&\cdots &k! \\
I^{k+1} & (k+1)I^{k} & (k+1)kI^{k-1} &  (k+1)k(k-1)I^{k-2}&\cdots &(k+1)k(k-1)\cdots 2I
\end{bmatrix}.
\label{Amatrix}\eeq
Consider the following elementary row operations (the equations will remain polynomial in nature)
\begin{itemize}
\item{\bf Step 1:} Divide row $i$ with $i$, for all $i$: $R_i/i$. 
\item{\bf Step 2:} Multiply row 1 with appropriate powers of $I$ and subtract row 1 from all other rows: $R_i\mapsto R_i-I^{i-1}R_1$, $\forall i$. This would turn column 2 into $(1,0,0,..,0)^T$. 
\item{\bf Step 3:} Repeat step 1 and 2 so that column 3 turns into $(0,1,0,...,0)$. 
\item{\bf Step $j$:} Repeat above steps so that column $j$ turns into $(0,...,0,1,0,...0)$. 
\end{itemize}
This algorithm ends with a simpler matrix representing the LHS
\beq
\begin{bmatrix}
I & 1 & 0 & 0 & \cdots & 0 \\
-\frac{I^2}{2} & 0 & 1 & 0 & \cdots & 0 \\
\frac{I^3}{3!} & 0 & 0 & 1 &\cdots & 0 \\
\vdots & \vdots & \vdots  & \vdots &\ddots & \vdots  \\
\frac{(-1)^{k-1}I^k}{k!} & 0 & 0 & 0&\cdots &1 \\
\frac{(-1)^{k}I^{k+1}}{(k+1)!}  & 0 & 0 & 0&\cdots & 0
\end{bmatrix};
\eeq
while the components of the RHS are polynomials in $I$. 
 Hence,  row $(k+1)$  represents the equation
\[ \frac{(-1)^k}{(k+1)!}\frac{D}{2} I^{k+1}=P_k(I),\]
where $P_k(I)$ is some polynomial with constant coefficients of order $\leq k$. 
Consequently, $I$ is constant (as well as all $X_i$'s). Since $I$ is arbitrary, we now have \CC$\Rightarrow$ CSI and the theorem follows. 
\end{proof} 
 
Given a CSI spacetime, $g_{ab}$, we can compute any polynomial invariant $J^i=j^i_0$ (=constant). If this invariant contains derivatives of the Riemann tensor up to order $p$, then we note that the corresponding conserved tensor obtained by varying $L_i=(J^i-j_0^i)^{m_i}$, $m_i=p+4$, is vanishing for $g_{ab}$ (corresponding to the last row in eq. (\ref{Amatrix})). Thus if $J^i,~i=1,...,N$ are (all of the) invariants, then the CSI spacetime will be a solution of the class of theories given by
\beq
L=\sum_{i=1}^Na_i(J^i-j_0^i)^{m_i}, \qquad a_i \text{ arbitrary.}
\eeq 

Here, we use the above theorem for the universal case where $T_{ab}=\lambda g_{ab}$. However, we should point out that above result applies to a bigger class of spacetimes; for example, the theorem  implies that the more general class of spacetimes for which all conserved symmetric rank-2 tensors $T_{ab}$ are covariantly constant; i.e., $T_{ab;c}=0$, also belong to the CSI class.


\section{The sufficiency and necessity proof for type N}
\label{sec_typeN}

In this section, we present the sufficiency and necessity proof {  of theorem \ref{prop_typeN}}. {  First in section \ref{secNsuff}, we 
prove that all type N Einstein Kundt spacetimes are universal. Let us briefly summarize the key points of the proof.

For the Weyl tensor itself, the proof is very simple. Any rank-2 tensor has only terms of boost weight $\geq$ -2 and type N  Weyl tensor only boost weight -2 terms.  Therefore all rank-2 tensors constructed from a type N Weyl tensor which are quadratic or of a higher order in the Weyl tensor vanish and, due to the tracelessness of the Weyl tensor, rank-2 tensors linear in the Weyl tensor vanish as well. Thus it is not possible to construct a non-vanishing rank-2 tensor from the type N Weyl tensor.

For covariant derivatives of the Weyl tensor, the proof is more involved. The key idea is to realize that, for type N Einstein { Kundt} spacetimes, an arbitrary covariant derivative of the Weyl tensor also admits only terms with boost weight -2 or less. This is certainly not true for non-Kundt type N spacetimes since the first covariant derivative contains b.w. -1 terms containing e.g.  $D {\Omega'}_{ij}$ and $\rho_{ij}$ (see \eqref{dl}--\eqref{dm}). However, for type N Einstein Kundt spacetimes, $\rho_{ij}=0$ and the Bianchi identities imply $D {\Omega'}_{ij}=0$ (see appendix \ref{App:NPeqs}).  To show that in this case the Bianchi and Ricci  identities also annihilate all terms with b.w. $\geq -1$ for all higher order derivatives of the Weyl tensor
 we use a generalization of
the balanced scalar approach introduced in \cite{Coleyetal04vsi}.  }

{  In section \ref{secNnecessity}, we then prove that type N universal spacetimes necessarily belong to the Kundt class.}

\subsection{The sufficiency proof}
\label{secNsuff}
{  Here we prove that   type  N Einstein Kundt spacetimes are universal.} 

We start with the following proposition:

\begin{prop}
\label{prop_1-balanced}
For type N Einstein  Kundt spacetimes, the boost order of  $\nabla^{(k)} C$  (a covariant derivative  of an arbitrary order of the Weyl tensor) 
{  with respect to the multiple WAND} is at most $-2$. 
\end{prop}
\begin{proof}
We  prove the above proposition using the balanced scalar approach introduced in \cite{Coleyetal04vsi}.
Let us  say that a  scalar $\eta$ with the boost weight b  is 1-balanced if $D^{-{\rm b}-1} \eta =0$ for ${\rm b}<-1$ and $\eta=0$ for ${\rm b} \geq -1$ and that a tensor is 1-balanced if 
all its  components are 1-balanced scalars. In the balanced scalar approach, it is convenient to assign boost weight to all NP quantities and thus  we consider only boosts with {\em constant} $\lambda$, then e.g. $L_{11}$ has b.w. -1.  

Now the  Ricci and Bianchi equations (see appendix \ref{App:NPeqs}) imply that for spin coefficients and Weyl components of various boost weights ${\rm b}$
\BEA
&& {\rm b}=-2:  \ \ D^3 \kappa'_i 
= 0, \ \ \ D {\Omega'}_{ij} =0,\label{DOmega}\\
&& {\rm b}=-1:  \ \ D^2 L_{11}=0, \ \ \ D^2 \rho'_{ij} 
= 0, \ \ \ D^2  \M{i}{j}{1} = 0, \label{DL11}\\
&& {\rm b}=0:  \ \ \ \ \ D {\tau_i} = 0 , \ \ \ D L_{1i}=0, \ \ \ D \M{i}{j}{k} =0.\label{DL_1i}
\EEA 

{  It follows that
\begin{lem}
\label{lem_1bal_coef}
In type N Einstein Kundt spacetimes, for a 1-balanced scalar $\eta$, scalars $L_{11} \eta$,  $\tau_i \eta$, $L_{1i}\eta$, $ \kappa'_i 
\eta $,  $ \rho'_{ij} 
\eta$,  $\Mi_{\!j1}\!  \eta$, $ \Mi_{\!kl}\!  \eta$ and $D \eta ,\ \delta_i \eta,\ \T \eta $ are also 1-balanced scalars.
\end{lem}
\begin{proof}
Let us present the proof for $\kappa'_i \eta$. The proof for scalars $L_{11} \eta$,  $\tau_i \eta$, $L_{1i}\eta$, 
 $ \rho'_{ij} 
\eta$,  $\Mi_{\!j1}\!  \eta$, $ \Mi_{\!kl}\!  \eta$  is similar.
From the definition, $D^q (\eta)=0$ for $q\geq -{\rm b}-1$  and from \eqref{DOmega}, $D^p (\kappa'_i)=0$ for $p\geq 3$. 
It follows that $D^p (\kappa'_i)D^q (\eta)=0$ if $p+q \geq -{\rm b}+1$. 
When expanding $D^{-{\rm b}+1}(\kappa'_i\eta)$ using the Leibnitz rule, all terms $D^p (\kappa'_i)D^q (\eta)$, $p+q=-{\rm b}+1$, vanish and thus 
$D^{-{\rm b}+1}(\kappa'_i\eta)=0$ and  $\kappa'_i\eta$ is 1-balanced.

For scalars involving derivatives, let us present the proof for $\delta_i \eta$. The proof for $\T \eta$ is similar and for $D\eta$ is trivial
($D^{-{\rm b}-2}(D\eta)=D^{-{\rm b}-1}\eta=0$). $D^q\eta$ is 1-balanced for all $q \geq 1$ and so is $L_{1i}D^q\eta$, i.e. $D^p (L_{1i}D^q \eta)=0$ for $p+q=-{\rm b}-1$. 
Using the commutator \eqref{comdD}
and the Leibnitz rule, it follows that  
$D^{-{\rm b}-1}(\delta_i\eta)=D^{-{\rm b}-2}(\delta_i D-L_{1i}D)\eta $, where $D^{-{\rm b}-2} (L_{1i}D \eta)$ vanishes. 
By repeatedly using the commutator \eqref{comdD} and using that $D^p (L_{1i}D^q \eta)=0$ 
for $p+q=-{\rm b}-1$, we arrive at 
 $D^{-{\rm b}-1}(\delta_i\eta)=\delta D^{-{\rm b}-1}\eta=0$ 
and thus $\delta_i\eta$ is 1-balanced.

\end{proof}}

In the covariant derivative of a 1-balanced tensor, only the above mentioned 1-balanced terms (like  $\tau_i \eta$, $D \eta_i$ etc.) appear and thus
 
\begin{lem}
\label{lem_1bal}
For type N Einstein  Kundt spacetimes, a covariant derivative of a 1-balanced tensor is a 1-balanced tensor.
\end{lem}

{  From \eqref{DOmega}, it follows that  for a type N Einstein Kundt spacetime, the Weyl tensor \eqref{typeNWeyl} is 1-balanced. 
Therefore, for these spacetimes, by lemma \ref{lem_1bal}, an arbitrary derivative of the Weyl tensor is also 1-balanced and by definition, all terms with boost weight
$\geq -1$ vanish. Proposition \ref{prop_1-balanced} immediately follows.} 
\end{proof}

{ 
Since  a rank-1 tensor has components of b.w. $\geq -1$ it is not possible to construct such a non-zero tensor from boost order -2 tensors. Therefore,
it follows from proposition \ref{prop_1-balanced} that 
\begin{cor}
\label{typeNcons}
For type N Einstein  Kundt spacetimes, all rank-2 tensors constructed from the Weyl tensor and its covariant derivatives of arbitrary order are conserved.
\end{cor} }

Let as now argue that proposition \ref{prop_1-balanced} implies that, for  type N Einstein  Kundt spacetimes, it is not possible to construct a 
non-vanishing rank-2 tensor from $\nabla^{(k)} C$ (and the metric).

First, note that any rank-2 tensor has components of boost weight -2 or higher. From proposition \ref{prop_1-balanced}, it follows that
\begin{lem}
\label{lemma_typeNquad}
 In  type N Einstein  Kundt spacetimes, rank-2 tensors constructed from $\nabla^{(k)} C$,  which are quadratic or of higher order in $\nabla^{(k)} C$, vanish.
\end{lem} 
We show below that  rank-2 tensors linear in $\nabla^{(k)} C$ vanish as well  {  (note that in order to construct such a rank-2 tensor, $k$ has to be even)}.  

For any tensor $T$, the commutator of covariant derivatives can be expressed as
\BE
[\nabla_a,\nabla_b]T_{c_1.\dots c_k}=T_{d\dots c_k}R^d_{c_1 ab}+\dots +T_{c_1\dots d}R^d_{c_kab} \label{commut_der}.
\EE

When $T$ is the  Weyl tensor, rank-2 tensors constructed from RHS of \eqref{commut_der} are either quadratic in $C$ (and thus vanish due to lemma 
\ref{lemma_typeNquad})  or contractions of the Weyl tensor with the Ricci tensor or the metric which, for Einstein spacetimes, vanish due to tracelessness of the Weyl tensor.

When expressing rank-2 tensors constructed from the Weyl tensor, we can thus assume that the first two covariant derivatives of the Weyl tensor commute. This will be expressed using the following notation
\BE
C_{abcd;ef} \cong C_{abcd;fe}. \label{commut2der}
\EE
Using \eqref{commut2der}, the Bianchi identities and tracelessness of the Weyl tensor, it follows that all rank-2 tensors constructed from $\nabla^{(2)} C$ vanish. 

Now, assume that for some even $n$ all covariant derivatives in $\nabla^{(n)} C $ in an expression of a rank-2 tensor commute and all rank-2 tensors {  constructed}
 from $\nabla^{(n)} C $ vanish. Using  
\eqref{commut_der}, we obtain
\be
C_{abcd;e_1\dots e_i f g h_1 \dots h_j} - C_{abcd;e_1\dots e_i g f  h_1 \dots h_j} \cong 0 \label{commutnp2der}
\ee
(where $i+j=n$) since when using \eqref{commutnp2der} in an expression of a rank-2 tensor {  constructed} from $\nabla^{(n+2)} C $ only rank-2 tensors {  constructed} from $\nabla^{(n)} C$ would appear on RHS of \eqref{commutnp2der}, but these are zero by our assumption. Using the Bianchi identities, it again follows that rank-2 tensors {  constructed} from $\nabla^{(n+2)} C $ vanish. Our assumption for $n$ thus also holds for $n+2$ and therefore these hold for all even $n$. Thus we arrived at

\begin{lem}
\label{lemma_typeNlinear}
 In  type N Einstein  Kundt spacetimes, rank-2 tensors constructed from $\nabla^{(k)} C$  which are linear in $\nabla^{(k)} C$ vanish.
\end{lem} 

A direct consequence of lemmas \ref{lemma_typeNquad} and  \ref{lemma_typeNlinear} is 
\begin{prop}
\label{prop_typeNuniv}
{ All type N Einstein Kundt spacetimes are universal.}
\end{prop}

\subsection{The necessity proof}
\label{secNnecessity}

Now, let us show that, for type N Einstein spacetimes, CSI implies Kundt.

The simplest non-trivial
curvature invariant  for type N spacetimes is  \cite{BicPra98} 
\BE
 I_{N} \equiv C^{a_1 b_1 a_2 b_2  ; c_1 c_2} C_{a_1 d_1 a_2  d_2 ; c_1 c_2} C^{e_1 d_1 e_2 d_2 ;f_1 f_2} C_{e_1 b_1 e_2 b_2  ; f_1 f_2} .
 \label{InvN}
\EE
In terms of higher dimensional GHP quantities, it can be shown \cite{Coleyetal04vsi} (see also  \cite{OrtPraPra10}) that $I_N$ is proportional (via a numerical constant) to  
\be
I_N\propto\left[ ({\Omega'}_{22})^2 + ({\Omega'}_{23})^2  \right]^2  (S^2+A^2)^4,
\ee
where $S$ and $A$ are introduced in \eqref{GStypeN}.
For type N Einstein spacetimes, $I_N$ is constant iff $S^2+A^2=0$, which implies Kundt. This can be shown either by expressing $D I_N$ and using the Bianchi and Ricci identities or by looking  at the explicit $r-$dependence of $I_N$ in the non-Kundt case \cite{OrtPraPra10} 
\be
I_N\propto \frac{\left[ ({\Omega'}_{22}^0)^2 + ({\Omega'}_{23}^0)^2  \right]^2 }{(r^2+a_0^2)^6},
\ee
where $r$ is an affine parameter along null geodetic integral curves of $\bl$.

We thus arrived at
\begin{lem}
CSI type N Einstein  spacetimes belong to the Kundt class.
\end{lem}
Together with theorem \ref{prop_univCSI} this gives
\begin{prop}
\label{prop_typeNuniv_Kundt}
Universal type N Einstein  spacetimes belong to the Kundt class.
\end{prop}

Theorem \ref{prop_typeN}   now directly follows from Propositions  \ref{prop_typeNuniv} and \ref{prop_typeNuniv_Kundt}.


\section{Necessary and sufficient conditions for type III spacetimes}
\label{sec_typeIII}

  In this section, we present the  proof { of theorem \ref{prop_typeIII}} - we show that $\tau_i=0$ type III Einstein Kundt spacetimes
obeying \eqref{QGterm} are universal. Let us here briefly summarize the main points of the proof.

First, we show that for  type III Einstein Kundt spacetimes (including both $\tau_i=0$ and  $\tau_i \not=0$ cases)  boost order of an arbitrary covariant derivative of the Weyl tensor with respect to the multiple WAND  is at most  -1. 
Thus, a non-vanishing rank-2 tensor constructed from the metric, the Weyl tensor and its covariant derivatives can be at most quadratic in the Weyl tensor or its derivatives. Then we use these results to prove the theorem \ref{prop_typeIII}.

Type III spacetimes are defined as spacetimes admitting a frame in which only  boost weight -1 and -2 components of the Weyl tensor are present.
{  In this frame, the Weyl tensor has the form } \cite{Coleyetal04,OrtPraPra12rev}
\be
   C_{abcd} = 8 {\Psi'}_{i}\, \ell^{}_{\{a} n^{}_b \ell^{}_c m^{(i)}_{\, d\: \}} +
    4 {\Psi'}_{ijk}\, \ell^{}_{\{a} m^{(i)}_{\, b} m^{(j)}_{\, c} m^{(k)}_{\, d\: \}} +
      4 {\Omega'}_{ij}\, \ell^{}_{\{a} m^{(i)}_{\, b}  \ell^{}_{c}  m^{(j)}_{\, d\: \}},\label{typeIIIWeyl}
\ee
{  where $ {\Psi'}_{i}=  {\Psi'}_{jij}$ and ${\Omega'}_{ij}$ is symmetric and traceless as in the type N case}.

{ 
As argued in \eqref{typeIIIconserved}, for type III Einstein spacetimes, $S^{\rm III}_{ab}$ is conserved and thus for type III universal spacetimes, $S^{\rm III}_{ab}$ is proportional to the metric and in fact, due to its vanishing trace,
\BE
S^{\rm III}_{ab} \equiv C_{acde} C_{b}^{\ cde} = 0. \label{typeIIIconservedvanishes}
\EE }

\subsection{The sufficiency proof}

\label{sec_typeIIIsuff}

We start with a type III analog of proposition \ref{prop_1-balanced}:
\begin{prop}
\label{prop-balanced}
For type III Einstein  Kundt spacetimes, the boost order of  $\nabla^{(k)} C$  (a covariant derivative  of an arbitrary order of the Weyl tensor)
{  {with respect to the multiple WAND}}  is at most $-1$. 
\end{prop}
\begin{proof}
The proof of this proposition is similar to the proof of proposition \ref{prop_1-balanced}.
In this case, we  define balanced scalars as  in \cite{Coleyetal04vsi}.
Let us  say that a  scalar $\eta$ with the boost weight ${\rm b}$  is balanced if $D^{-{\rm b}} \eta =0$ for ${\rm b}<0$ and $\eta=0$ for ${\rm b} \geq 0$ and that a tensor is balanced if all its  components are balanced scalars.

For type III Einstein Kundt spacetimes, the Ricci and Bianchi equations (see appendix \ref{App:NPeqs})  for spin coefficients and Weyl components of various boost weights ${\rm b}$ imply that
\BEA
&& {\rm b}=-2:  \ \ D^3 \kappa'_i 
                  = 0, \ \ \ D^2 {\Omega'}_{ij} =0,\\
&& {\rm b}=-1:  \ \ D^2 L_{11}=0, \ \ \ D^2 \rho'_{ij} 
= 0, \ \ \, \ D^2  \M{i}{j}{1} = 0, \ \ \ D {\Psi'}_{ijk} =0, \ \ \  D {\Psi'}_{i} = 0, \\
&& {\rm b}=0:  \ \ \ \ \ D {\tau_i} = 0 , \ \ \ D \M{i}{j}{k} =0.
\EEA 
It follows that for a balanced scalar $\eta$, scalars $L_{11} \eta$,  $\tau_i \eta$, $ \kappa'_i 
\eta $,  $ \rho'_{ij} 
\eta$,  $\Mi_{\!j1}\!  \eta$ and $ \Mi_{\!kl}\!  \eta$ are also balanced scalars. Using commutators given in appendix \ref{App:NPeqs}, one can  show that 
$D \eta ,\ \delta_i \eta,\ \T \eta $ are  balanced scalars as well.

Analogously to the type N case, it follows that a covariant derivative of a balanced tensor is a balanced tensor and since type III Weyl tensor for a  type III Einstein Kundt spacetime  is balanced, proposition \ref{prop-balanced} follows. 
\end{proof}

{  From proposition \ref{prop-balanced}, it follows that for  type III Ricci-flat Kundt spacetimes, a rank-2  tensor that has in general components only with boost weight $\geq -2$, can be  at most quadratic in the Weyl tensor and its covariant derivatives since components of the Weyl tensor and its derivatives have boost weight $\leq -1$. }
{  Thus,}
 a direct consequence of proposition \ref{prop-balanced} is

\begin{lem}
\label{lemma_quadraticC}
For  type III Ricci-flat Kundt spacetimes, a non-vanishing rank-2  tensor constructed from the metric, the Weyl tensor and its covariant derivatives of arbitrary order is at most quadratic in the Weyl tensor and its covariant derivatives.
\end{lem}

{ 
Note that for type III Ricci-flat Kundt spacetimes, following a  similar argument as in corollary \ref{typeNcons}, all rank-2 tensors quadratic in the Weyl tensor and its derivatives are conserved. }

However, note that the rank-2 Weyl polynomial $C_{acde} C_{b}^{\ cde}$  which is  quadratic in $C$ is in general non-vanishing for type III Einstein Kundt spacetimes \cite{MalekPravdaQG} and thus in general these spacetimes are not universal.
Obviously, type III Kundt universal spacetimes will be subject to  \eqref{QGterm}.
Furthermore, even when assuming that \eqref{QGterm} holds, we find that in the Ricci-flat case, where FKWC basis \cite{FKWC} of rank-2, order-6 Weyl polynomials consists of 
\BE
F_1 \equiv C^{pqrs} C_{pqrs;ab}, \ \ \ F_2 \equiv C^{pqrs}_{\ \ \ \ \ ;a} C_{pqrs;b},  \ \ \ F_3 \equiv C^{pqr}_{\ \ \ \ a;s} C_{pqrb}^{\ \ \ \ \ ;s}, \label{eqF}
\EE
$F_1$ and $F_2$ are in general non-vanishing (while $F_3$ vanishes as a consequence of  \eqref{QGterm}). We  show in section \ref{subsec_quad} that, 
 for $\tau_i =0$, $F_1$ and $F_2$ vanish.
Thus in the following, we   focus on the $\tau_i =0$ case leaving the question whether some  particular type III Kundt universal metrics with $\tau_i \not=0$ exist open.

 
\subsubsection{Rank-2 Weyl polynomials quadratic in $\nabla^{(k)} C$ for type III,  $\tau_i =0$ Einstein  Kundt spacetimes obeying \eqref{QGterm} vanish }
\label{subsec_quad}

For type III Einstein Kundt, the Ricci identity \eqref{Ricci_11i} implies that $\Lambda$ vanishes for $\tau_i =0$. 

From the commutator of partial derivatives  \eqref{commut_der} and lemma \ref{lemma_quadraticC}, it follows that in an expression for a rank-2 Weyl polynomial
quadratic in $\nabla^{(k)} C$, RHS of \eqref{commut_der} vanishes (since all terms are cubic in $C$) and thus   covariant derivatives in such expressions effectively commute
\BE
T_{c_1.\dots c_k;ab} \cong T_{c_1.\dots c_k;ba}. \label{commut}
\EE
This allows us to prove\begin{lem}
\label{lemma_insum}
For  type III, $\tau_i=0$ Ricci-flat Kundt spacetimes, a rank-2  tensor constructed from the metric, the Weyl tensor and its covariant derivatives of arbitrary order quadratic in $\nabla^{(k)} C$ vanishes if it contains a summation within $\nabla^{(k)} C$. 
\end{lem}
\begin{proof}
i) A summation within $C$ vanishes due to the tracelessness of the Weyl tensor. ii) A summation within $\nabla C$ can be (using the Bianchi identities) expressed in terms of traces of the Weyl tensor and thus it again vanishes. iii) Using the Bianchi identities \eqref{commut} and result of ii), it follows that a summation within $\nabla^{(2)} C$ vanishes. Similar argument holds for higher derivatives.   
\end{proof}

A direct consequence is 
\begin{lem}
\label{lemma_der}
For  type III, $\tau_i=0$ Ricci-flat Kundt spacetimes, let us assume that a certain rank-2 polynomial quadratic in $\nabla^{(k)} C$ vanishes. Symbolically we will write $C^{(1)} C^{(2)} =0$. Then also $C^{(1)}_{\ ;f} C^{(2)\ \! f}_{\ \  ;} =0$.
\end{lem}

From proposition \ref{prop-balanced}, we know that  $\nabla^{(k)} C$ has only terms of b.w. $\leq-1$. Now, let us look at boost weight $-1$ terms after one further differentiation, i.e., b.w. $-1$ terms in $(\nabla^{(k)} C_{abcd})_{;e}$. Such terms arise in one of the three following ways:\\
i) From \eqref{dl} - \eqref{dm} and \eqref{eqnabla} for $\tau_i=L_{1i}=\rho_{ij}=0$, we see that
  the differentiations of $\bl's$ and $\bn's$ in the frame decomposition of $\nabla^{(k)} C_{abcd}$ do not lead to b.w. $-1$ terms, while 
 the differentiation of vectors $\mbox{\boldmath{$m^{(i)}$}} $ leads to  b.w. $-1$ terms (via term ${\Mi}_{kl} m^{(k)}_{\, a} m^{(l)}_{\, e}$ in \eqref{dm}). \\
ii) b.w. -1 terms also arise from the differentiation of b.w. $- 1$ frame components (symbolically denoted as $\eta_{(-1)}$) of $\nabla^{(k)} C$ via 
${(\eta_{(-1)}) }_{;e}\ \rightarrow \ \delta_i (\eta_{(-1)})  m^{(i)}_{e} $. \\
iii) Finally b.w. -1 terms also arise from the differentiation of b.w. $- 2$ frame components $\eta_{(-2)}$ via 
${(\eta_{(-2)}) }_{;e}\ \rightarrow \ D (\eta_{(-2)})  n_{e} $.

It follows, that b.w. $-1$ terms of $(\nabla^{(k)} C_{abcd})_{;e}$ do not contain vector $\ell_e$ {  (i.e. $(\nabla^{(k)} C_{abcd})_{;e}n^e=0$)}. 

Since any rank-2 tensor constructed from the metric, the Weyl tensor and its derivatives has the form $S_{ab} = f \ell_a \ell_b$ and only b.w. $-1$ terms in $\nabla^{(k)} C$ can contribute to the result,
 it follows (employing also \eqref{commut}) that
\begin{lem}
\label{lemma_freeind}
For type III, $\tau_i=0$ Ricci-flat  Kundt spacetimes,  any rank-2 tensor constructed from the metric, the Weyl tensor and its derivatives, quadratic in $\nabla^{(k)} C$ with at least one free derivative index vanishes.  
\end{lem}

This immediately implies that $F_1$ and $F_2$ given in \eqref{eqF} vanish and since the remaining rank-2 polynomial in order-6 FKWC basis, $F_3$,  also vanishes due to \eqref{QGterm} and {  lemma \ref{lemma_der},} it follows that all order-6, rank-2 polynomials quadratic in $\nabla^{(k)} C$ vanish as well.

Taking into account the Weyl tensor symmetries,  lemmas \ref{lemma_insum}, \ref{lemma_der}, \ref{lemma_freeind} and eq. \eqref{commut},  it follows that an existence of a non-vanishing  rank-2 tensor constructed from the metric, the Weyl tensor and its derivatives and quadratic in $\nabla^{(k)} C$ implies an existence of such a 
non-vanishing  rank-2 tensor
constructed from at most second derivatives of the Weyl tensor. 


{ 
For  type III, $\tau_i=0$ Ricci-flat Kundt spacetimes, let us now  try to find 
a non-vanishing  rank-2  tensor $K_{ab}$ constructed from the metric, the Weyl tensor and its derivatives, quadratic in $\nabla^{(k)}C $. 

From lemma \ref{lemma_freeind}, we know that all derivative indices have to be  dummy indices and that 
(by lemma \ref{lemma_insum}) these indices cannot be summed within $\nabla^{(k)} C$. Since in the expression 
of  $K_{ab}$, the derivative indices effectively commute  \eqref{commut} we can ignore the dummy indices which are both 
derivative indices since
 an  existence of a non-vanishing  $K_{ab}$ of
the form $C^{(1)}_{\ ;f} C^{(2)}_{\ ;f}$ would imply (by lemma \ref{lemma_der}) an existence of another non-vanishing $\tilde K_{ab}$ of the form
$C^{(1)} C^{(2)}$. Since derivative indices in $C^{(1)}$ can be summed only with non-derivative indices in $C^{(2)}$ and vice-versa, it follows that 
$C^{(1)}$ and $C^{(2)}$ are at most second derivatives of  the Weyl tensor. However, as argued above, all the 6th order rank-2 tensors vanish and thus  
 $C^{(1)}$ and $C^{(2)}$ are precisely second derivatives of the Weyl tensor.

Taking into account all above mentioned constraints, we see that the only potentially non-vanishing 8th order rank 2 tensor is (modulo trivial identities) }
$$
C_{agdh;ef} C_ b^{\ edf;gh},
$$
which, however, vanishes as well due to the Bianchi identities and lemma \ref{lemma_freeind}. 

We can thus conclude this section with
\begin{prop}
\label{prop_quad}
For  type III Ricci-flat Kundt spacetimes obeying  $C_{acde} C_{b}^{\ cde}=0$ and $\tau_i=0$, all rank-2  tensors constructed from the metric, the Weyl tensor and its covariant derivatives of arbitrary order which are  quadratic or of higher order in the Weyl tensor and its derivatives vanish.
\end{prop}

It remains to show that, under assumptions of proposition \ref{prop_quad}, rank-2 Weyl polynomials linear in $\nabla^{(k)} C$ vanish as well.

\subsubsection{Rank-2 Weyl polynomials linear in $\nabla^{(k)} C$}

From proposition \ref{prop_quad} and eq. \eqref{commut_der}, it follows that when expressing rank-2 polynomials linear in $\nabla^{(k)} C$, covariant derivatives commute.
Vanishing of rank-2 Weyl polynomials linear in $\nabla^{(k)} C$ is then a trivial consequence of the Bianchi identities and the tracelessness of the Weyl tensor.

Together with the results given in section \ref{sec_typeIIIsuff} above, this completes the proof of theorem \ref{prop_typeIII}.

\subsection{Necessary conditions}
\label{sec_necessaryIII}

In has been shown in \cite{Pravdaetal04}, that under some additional genericity assumptions, the optical matrix of type III Einstein spacetimes has the form \eqref{GStypeN}. In particular, this holds for all five-dimensional type III Einstein spacetimes and all non-twisting spacetimes in arbitrary dimensions. In fact, type III Einstein spacetimes not obeying \eqref{GStypeN} are not known.

 The simplest non-trivial curvature invariant for  type III Einstein spacetimes is  \cite{Coleyetal04vsi} 
\BE
I_{{III}} = C^{a_1 b_1 a_2 b_2;e_1} C_{a_1 c_1 a_2 c_2;e_1} C^{d_1 c_1 d_2 c_2;e_2} C_{d_1 b_1 d_2 b_2;e_2}.\label{invIII}
\EE
Assuming that eq. \eqref{GStypeN} holds,  the r-dependence of  $I_{{III}}$ was determined in eq. (97) of \cite{OrtPraPra10}.  It follows that, for non-Kundt spacetimes,
$I_{{III}}$ is clearly non-constant. 

Therefore, universal type III Einstein spacetimes obeying \eqref{GStypeN} belong to the Kundt class.


 \section{The Kundt metrics}
\label{sec_Kundt}

As follows from the results discussed above, all universal metrics discussed here (and in fact possibly all universal metrics in general) belong to the Kundt class.
{  In order to give illustrative {\em explicit} examples of universal spacetimes, let us  briefly present the relevant subclasses of Kundt metrics.}
{  Kundt metrics are} defined as spacetimes admitting a null geodetic vector field $\bl$ with vanishing shear, expansion and twist, {i.e., $\kappa_i=0=\rho_{ij}$}.
{  The Kundt vector field $\bl$ can be rescaled by a boost so that in appropriate coordinates $\bl={\rm d}u$ (c.f. \eqref{Kundt_gen})}, {  $L_{1i}=\tau_i$ in \eqref{dl} and thus} the covariant derivative  of $\bl$ simplifies to  \cite{PodOrt06,PodZof09}
\be
\ell_{a;b}=L_{11}\ell_a\ell_b+\tau_i(\ell_a m^{(i)}_b+m^{(i)}_a\ell_b). 
\ee

The  type N Ricci-flat Kundt class splits into two subclasses characterized by vanishing/non-vanishing $\tau_i$ 
{  (see \cite{OrtPraPra12rev,Coleyetal06,PodZof09,Coleyetal03})}. It can be shown that for  $\tau_i=0$ class, {  from the Ricci identities \eqref{Ricci_DL},\eqref{Ricci_TL}, it follows that $L_{11}$ depends only on the coordinate $u$} and one can rescale $\bl$ {  by a $u$-dependent boost} to set $L_{11}=0$. This class thus represents type N Ricci-flat pp waves. 

The type III 
Ricci-flat  Kundt class  splits again in two subclasses characterized by vanishing/non-vanishing $\tau_i$. 
In the $\tau_i=0$ class,  representing recurrent spacetimes,  {  the Ricci identities \eqref{Ricci_DL}, \eqref{Ricci_TL}} imply that {  $L_{11}=L_{11}(u)$ and therefore $L_{11}$} { can} be set to zero {  by a  $u$-dependent boost} iff ${\Psi'}_i=0$. 
Thus, for $\tau_i=0={\Psi'}_i$, this class represents type III Ricci-flat pp waves.

The condition \eqref{QGterm}  is an identity for type N and III spacetimes in four dimensions.
Furthermore, if the condition \eqref{QGterm} holds for an n-dimensional spacetime it is also valid for $n+1$ dimensional spacetime obtained by warping the original metric
\cite{OrtPraPra11}.
Since also an algebraic type of an algebraically special spacetime is preserved under warp \cite{OrtPraPra11}, one can easily construct higher dimensional type III Kundt spacetimes
 obeying  \eqref{QGterm} by warping four-dimensional type III Kundt spacetimes.
 For $n>4$ type III, the condition \eqref{QGterm}  can be expressed as  \cite{MalekPravdaQG}
\be
C_{acde} C_{b}^{\ cde} = \tilde \Psi  \ell_a \ell_b = 0,
 \label{QG_III}
\ee 
where
\be
\tilde \Psi \equiv  \frac{1}{2} \Psi'_{ijk} \Psi'_{ijk} -  \Psi'_i \Psi'_i.
 \label{tPsi}
\ee
It is clear that, for type III pp waves, $\tilde \Psi \not= 0$  (since $\Psi'_i=0$, see \eqref{Ricci_TL}, while $\Psi'_{ijk} \not=0$) and thus these spacetimes are not universal.

Note that, for type III and N Einstein Kundt spacetimes with $\tau_i=0$, cosmological constant $\Lambda$ necessarily vanishes due to the Ricci identity 
\eqref{Ricci_11i}, while the $\tau_i \not= 0$ class allows for non-vanishing $\Lambda$.

For the Kundt class, one can introduce coordinates such that $\bl=\partial_r$ and $\ell_a\d x^a=\d u$ and the metric takes the form \cite{Coleyetal03,ColHerPel06} (see also more recent papers studying the Kundt class \cite{Coleyetal09,PodZof09,PodolskySvarc2012})
\be
 \d s^2 =2\d u\left[\d r+H(u,r,x^{\gamma})\d u+W_\alpha(u,r,x^{\gamma})\d x^\alpha\right]+ g_{\alpha\beta}(u,x^{\gamma}) \d x^\alpha\d x^\beta , \label{Kundt_gen}
\ee
where {  the coordinates are $\{ u, r, x^\alpha\}$, with $\alpha, \beta, \gamma =2 \dots n-1$.} 
The remaining frame vectors can be chosen {  in the following way} $n_a\d x^a=\d r+H\d u+W_\alpha\d x^\alpha$ and $m^{(i)}_a\d x^a=e^{(i)}_\alpha\d x^\alpha$ with 
$g_{\alpha\beta}= \delta_{ij}e^{(i)}_\alpha e^{(j)}_\beta$. It follows that $\ell_{a;b}=\frac{1}{2} g_{ab,r}$, $L_{11}=H_{,r}$ and $\tau_i=\frac{1}{2} W_{\alpha},_{r} e^{\alpha}_{(i)}=\frac{1}{2} W_{i},_{r}$. Obviously, for $g_{ab}$ independent on $r$, $\bl$ is a CCNV.

Since, by theorem \ref{prop_univCSI}, universality implies CSI, we can restrict ourselves to the {\em Kundt CSI metrics}, where  \cite{ColHerPel06,Coleyetal09}
 \BEA
W_{\alpha}(u,r,x^{\gamma})&=&rW^{(1)}_\alpha (u,x^{\gamma}) + W^{(0)}_\alpha (u,x^{\gamma}), \nonumber \\
H(u,r,x^{\gamma})& = &\frac{r^2}{8}\left(a+W^{(1)}_\alpha W^{(1)\alpha}\right)+ r H^{(1)} (u,x^{\gamma}) + H^{(0)} (u,x^{\gamma}), \label{CSI_Kundt_2}
\EEA
$g_{\alpha\beta}(x^{\gamma})$ (note that $g_{\alpha\beta,u}=0$) is {  the metric of} a locally homogeneous space and $a$ is a constant. Note that \eqref{CSI_Kundt_2} are necessary but not sufficient conditions for Kundt CSI. Let us now further specialize to the Ricci-flat case.

\subsection{Ricci-flat type III and N Kundt spacetimes}

{  A well known subclass of Kundt spacetimes are VSI spacetimes (spacetimes with all curvature invariants vanishing) of \cite{Coleyetal04vsi}.
Although the VSI property is not necessary nor sufficient condition for universality  all type N Ricci-flat VSI spacetimes  and $\tau_i=0$, type III Ricci-flat VSI spacetimes are universal.}
{  Therefore, here, we briefly  present  the VSI metrics, studied in detail in \cite{Coleyetal06}. For the VSI spacetimes, the metric \eqref{Kundt_gen} reduces to} 
\be
 \d s^2 =2\d u\left[\d r+H(u,r,x^{\gamma})\d u+W_\alpha(u,r,x^{\gamma})\d x^\alpha\right]+ \delta_{\alpha\beta} \d x^\alpha\d x^\beta , \label{Kundt_VSI}
\ee
with
\BEA
W_2 &=& - \frac{2 \epsilon r}{x^2}, \label{VSIW2} \\ 
W_{M}(u,r,x^{\gamma})&=& W^{(0)}_M (u,x^{\gamma}), \label{VSIWM} \\
H(u,r,x^{\gamma})& = &\frac{\epsilon r^2}{2 (x^2)^2}+ r H^{(1)} (u,x^{\gamma}) + H^{(0)} (u,x^{\gamma}) \qquad (\epsilon=0,\ 1), \label{VSIH}
\EEA
which for type N, reduce to
\BEA
W_2 &=& - \frac{2 \epsilon r}{x^2}, \label{typeNW2} \\ 
W_M &=& x^N B_{NM}(u)+C_M(u)[x^2+\epsilon (1-x^2)],\\
H   &=& \frac{\epsilon r^2}{2 (x^2)^2}+  H^{(0)} (u,x^{\gamma}) \label{typeNH},
\EEA
where $B_{NM}=B_{[NM]}$,  $M,N=3 \dots n-1$, $\epsilon = 0 $ or 1 and the $\tau_i = 0$ case corresponds to $\epsilon = 0 $.

Universal metrics  of \cite{HorSte90}  are of type N with $\tau_i=0$ and they
 correspond to the subset of  metrics \eqref{Kundt_VSI}, \eqref{typeNW2} - \eqref{typeNH} with $\epsilon=0$, $B_{NM}(u)=0$, $C_M(u)=0$. 
{  However, note that in four dimensions, the class of metrics considered in \cite{HorSte90}  and the $\epsilon=0$ subset of \eqref{Kundt_VSI}, \eqref{typeNW2} - \eqref{typeNH}  coincide since  for $n=4$ and $\epsilon=0$, $B_{NM}(u)$ vanishes and $C_M(u)$ can be transformed away.}

Further conditions follow from the Ricci-flat condition (see \cite{Coleyetal06}). 

\subsection{Explicit examples of universal metrics}

{  Now, let us present a few explicit examples of universal metrics of type N and III.}

\subsubsection{Type N Einstein Kundt}
In four dimensions, after a coordinate transformation of the form $r=v Q^2/P^2$, all type N Einstein  Kundt metrics can be  expressed as \cite{Ozsvathetal85,BiPo98,Gripobook}
\BE
 \d s^2 = 2 \frac{Q^2}{P^2} \d u \d v + \left( 2 k \frac{Q^2}{P^2} v^2  +  \frac{\left(Q^2\right)_{,u}}{P^2} v - \frac{Q}{P} H \right)  \d u^2
  + \frac{1}{P^2} \left( \d x^2 + \d y^2 \right)
  , \label{KN}
\EE
where 
\BEA 
P=1+\frac{ \Lambda}{12} (x^2+y^2) , \quad
k=\frac{  \Lambda}{6} \alpha(u)^2+\frac{1}{2} \left(\beta(u)^2+\gamma(u)^2 \right) , \nonumber \\
Q=\left(1-\frac{  \Lambda}{12} (x^2+y^2) \right) \alpha(u)+\beta(u) x+\gamma(u) y , \ \ 
H = 2 f_{1,x} - \frac{  \Lambda}{3 P} (x f_1 + y f_2) ,\nonumber
\EEA 
where $\alpha(u), \beta(u), \gamma(u)$ are free functions (see \cite{Gripobook} for canonical forms) and 
$f_1=f_1(u,x,y)$ and $f_2=f_2(u,x,y)$ obey $f_{1,x}=f_{2,y}$,  $f_{1,y}=-f_{2,x}$. 

Higher dimensional type N Einstein Kundt spacetimes can be generated by warping the metric \eqref{KN}. Note that, in general, singularities appear as a consequence of the warped product unless cosmological constants of both ($n$ and $n+1$ dimensional) solutions are both negative or both zero (see \cite{OrtPraPra11} for more details). 

Other examples in this class are (A)dS-waves of various kinds. For example, 
\be
 \d s^2=e^{-pw}\left(2\d u\d v+H(u,w,x^M)\d u^2+\delta_{MN}\d x^M \d x^N \right)+\d w^2, 
\ee
with $p$ being a constant and $H$ obeying $H,_{KL}\delta^{KL}+\left(H,_{ww}-\frac{n-1}{2}pH,_w \right) {\rm e}^{-pw}=0$
and 
\be
  \d s^2=\sinh^2(p w)\left[2\d u\d v+\left(v^2p^2+H(u,w,x^M)\right)\d u^2\right]+\d w^2+\frac 1{p^2}\cosh^2(p w)\d S^2_{{\mathbb H}^d}, 
\ee
where $\d S^2_{{\mathbb H}^d} = \Omega^{-2}\delta_{MN}\d x^M\d x^N$ ($\Omega=1-\frac{1}{4}\delta_{KL}x^Kx^L$)
is the `unit' metric on the $d$-dimensional hyperbolic space, and $H$ satisfying 
\be
p^2\Omega\left[ 2\Omega H,_{MN}\delta^{MN}+(d-2)H,_{M}x^M\right]
+2 c^2\left[ H,_{ww}+pH,_w \left( 2\frac{c}{s}+d\frac{s}{c}\right)\right]=0,
\ee
where $c=\ch$ and $s=\sh$.

\subsubsection{Type III, $\tau_i=0$, Ricci-flat Kundt spacetimes}

An explicit example of a four-dimensional type III, $\tau_i=0$ Ricci-flat Kundt universal metric (expressed in other coordinates) is \cite{petrovart1}
\be
\d s^2 = 2 \d u \d v - x (v + \mathrm{e}^{x}) \d u^2  + \mathrm{e}^{x} (\d x^2 + \mathrm{e}^{-2u} \d y^2). \label{petrovmetric}
\ee
One can obtain a higher dimensional type III universal metric as a direct product of \eqref{petrovmetric} with extra flat dimensions
since the algebraic type and the conditions \eqref{QGterm} and $\tau_i=0$ will be preserved.

\section{Conclusion}
\label{sec_concl}

Universal spacetimes are vacuum solutions to {all} theories with {  the Lagrangian depending on } 
the metric, the Riemann tensor and its derivatives of arbitrary order. 
In this work, we  study necessary and sufficient conditions for such spacetimes of arbitrary dimension and relate universal spacetimes to other known classes of metrics.
First, theorem \ref{prop_univCSI} states   that the  necessary (but not sufficient) condition  for universality of a spacetime is to be a CSI spacetime.

Then, we focus on type N and III spacetimes.
 For type N spacetimes, 
we find  simple necessary and sufficient conditions for universality. In this case, theorem \ref{prop_typeN} states that a type N spacetime is universal if and only if it is an Einstein Kundt spacetime. This class of spacetimes consists of pp waves admitting CCNV $(\tau_i =0,\Lambda=0)$ and Kundt waves $(\tau_i \not= 0, \Lambda$ arbitrary$)$. 
A four-dimensional metric for these spacetimes is given in eq. \eqref{KN}. In higher dimensions, a metric has the form \eqref{Kundt_gen}, which simplifies to \eqref{Kundt_VSI}, \eqref{typeNW2}-\eqref{typeNH}
 in the Ricci-flat case.

For type III, conditions for universality are more complicated. We  show (theorem \ref{prop_typeIII}) that  type III, $\tau_i =0$ Einstein  Kundt spacetimes obeying  \eqref{QGterm} 
are universal.  The condition \eqref{QGterm} 
is a necessary condition for universality of a type III spacetime.   The question whether the Kundt condition is also a necessary condition for universality
in the type III case is answered only under some genericity assumptions.  
Although type III,  $\tau_i \not=0$ Einstein Kundt spacetimes
are not in general universal,  necessity of $\tau_i =0$ condition also remains  open.

For dimensions $n>4$,  the type III, $\tau_i =0$ Einstein Kundt class contains also type III pp waves. Interestingly, these pp waves are {\em not} universal. Instead, type III universal spacetimes {  discussed above} admit a recurrent null vector.  
The metric of type III universal spacetimes has the form \eqref{Kundt_VSI} - \eqref{VSIH} with $\epsilon=0$. 

Although we have not discussed universal spacetimes of type II here, there are a few examples of universal spacetimes known already \cite{Coleyetal08}; for example, 
{  a Kundt metric of the form} 
\beq
\d s^2=2\d u\d v+\left(-v^2\lambda+H(u,x,y)\right)\d u^2+\frac {1}{\lambda}(\d x^2+\sinh^2 x\d y^2), \quad \Box H=0.
\eeq
It is also believed that {  all} type II universal  spacetimes are Kundt. We  leave the investigation of these spacetimes for a future work.

\section*{Acknowledgments}
V.P. and A.P. acknowledge support from research plan RVO: 67985840 and research
grant GA\v CR 13-10042S. We are grateful to Marcello Ortaggio and Lode Wylleman for useful comments on the draft.

\appendix

\section{The Ricci and Bianchi equations and commutators for type III and N  Einstein Kundt spacetimes in a parallelly propagated frame.}
\label{App:NPeqs}

Throughout the paper, we repeatedly use commutators and a subset of the Ricci and Bianchi equations for type III and N  Einstein Kundt spacetimes in a parallelly propagated frame. For convenience, we list these equations here. The original, more general form can be found in 
\cite{Pravdaetal04, Coleyetal04,OrtPraPra07,Durkeeetal10}.   

The Ricci equations \cite{OrtPraPra07,Durkeeetal10}:
\BEA
DL_{11} & =&  - L_{1i} \tau_{i}  +{\cal R} ,\label{Ricci_DL}\\
 DL_{1i} & = &0 , \\
\bigtriangleup L_{1i} - \delta_i L_{11}& = &   L_{11}(L_{1i}- \tau_{i}) -  \tau_{j} N_{ji} -L_{1j} (N_{ji}+\M{j}{i}{1}) + \Psi'_i , \label{Ricci_TL}\\
 D \tau_{i} & = &0 , \\
 D \kappa'_i 
&=& - \rho'_{ij} \tau_{j} + \Psi'_i  , \\
 - \delta_j \tau_{i} &=&  - \tau_{i} \tau_{j}   +  \tau_{k} \M{k}{i}{j} 
-{\cal R}\delta_{ij},  \label{Ricci_11i}\\
  D \rho'_{ij} 
	& = & -{\cal R}\delta_{ij}	, \\
  D \M{i}{j}{1} & =& - \M{i}{j}{k} \tau_{k},  \\
  D \M{i}{j}{k} &=&  0 ,   \\
\EEA
where 
\be
{\cal R}\delta_{ij}=\textstyle{\frac{1}{n-2}} (R_{ij}+R_{01} \delta_{ij})  - \textstyle{\frac{1}{(n-1)(n-2)}} R \delta_{ij}=\frac{R}{n(n-1)} \delta_{ij}\nonumber
\ee
for Einstein spaces $R_{ab}=(R/n) g_{ab}$.

The Bianchi equations \cite{Pravdaetal04,Durkeeetal10}:

\BEA
D \Psi'_i&=&0,\\
D\Psi'_{ijk}&=&0,\\
D \Omega'_{ij}&=&\delta_j\Psi'_i+\Psi'_i(L_{1j}-\tau_j)+\Psi'_{ijs}\tau_{s}+\Psi'_s \M{s}{i}{j} .
\EEA

The commutators \cite{Coleyetal04}:
\BEA
\T D - D \T &=& L_{11} D + \tau_i 
\delta_i,  \label{comTD} \\
\delta_i D - D \delta_i &=& L_{1i} D , \label{comdD} \\
\delta_i \T - \T \delta_i  &=& \kappa'_{i} 
D + (\tau_i 
-L_{1i}) \T + (\rho'_{ji} 
-\Mi_{j1}) \delta_j ,\\
\delta_i \delta_j - \delta_j \delta_i &=& (\rho'_{ij} 
-\rho'_{ji} 
) D 
 + (\Mk_{ij}-\Mk_{ji})\delta_k.
\EEA


\end{document}